\documentclass[12pt]{article}
\usepackage{amssymb}
\usepackage{amsmath}
\usepackage{amsthm}
\usepackage{bm} 
\usepackage{graphics}
\usepackage{graphicx}
\usepackage[T1]{fontenc}
\usepackage{longtable}
\usepackage{natbib}
\usepackage{setspace}
\usepackage{times}

\usepackage{epsfig}
\usepackage{fancyvrb}
\usepackage{multirow}
\usepackage[affil-it]{authblk}
\usepackage{verbatim}
\usepackage{tablefootnote}
\usepackage[explicit]{titlesec}
\usepackage{dcolumn}
\usepackage[algo2e]{algorithm2e}
\usepackage{gensymb}

\usepackage{tikz}
\usetikzlibrary{shapes,snakes, graphs}

\usetikzlibrary{arrows,%
                petri,%
                topaths}%
\usepackage{tkz-berge}
\usepackage{pgf}
\usepackage{subfigure,tikz}
\usetikzlibrary{arrows,automata}

\newtheorem{Theorem}{Theorem}
\newtheorem*{Theorem*}{Theorem}

\newtheorem{Assumption}{Assumption}

\def\ci{\mbox{\ensuremath{\perp\!\!\!\perp}}}
\def\nci{\not\!\perp\!\!\!\perp}
\def\E{\mathbb{E}}

\DeclareMathOperator*{\argmin}{arg\,min}
\SetKwRepeat{Do}{do}{while}

\titleformat{\section}
  {\normalfont\Large\bfseries\centering}{\thesection.}{1em}{\MakeUppercase{#1}}
  
\begin{document}

\def\spacingset#1{\renewcommand{\baselinestretch}%
{#1}\small\normalsize} \spacingset{1}

\title{\bf A Class of Semiparametric Tests of Treatment Effect Robust to Confounder Classical Measurement Error}
\author{Caleb H. Miles, Joel Schwartz, and Eric J. Tchetgen Tchetgen
\thanks{Caleb H. Miles is Postdoctoral Fellow, Department of Biostatistics, University of California, Berkeley 94720-7358. Joel Schwartz is Professor, Departments of Environmental Health and Epidemiology, Harvard T.H. Chan School of Public Health, Boston, MA 02115. Eric J. Tchetgen Tchetgen is Professor, Departments of Biostatistics and Epidemiology, Harvard T.H. Chan School of Public Health, Boston, MA 02115.}}
\date{}
\maketitle
\bigskip

\begin{abstract}
\noindent When assessing the presence of an exposure causal effect on a given outcome, it is well known that classical measurement error of the exposure can reduce the power of a test of the null hypothesis in question, although its type I error rate will generally remain at the nominal level. In contrast, classical measurement error of a confounder can inflate the type I error rate of a test of treatment effect. In this paper, we develop a large class of semiparametric test statistics of an exposure causal effect, which are completely robust to classical measurement error of a subset of confounders. A unique and appealing feature of our proposed methods is that they require no external information such as validation data or replicates of error-prone confounders. We present a doubly-robust form of this test that requires only one of two models to be correctly specified for the resulting test statistic to have correct type I error rate. We demonstrate validity and power within our class of test statistics through simulation studies. We apply the methods to a multi-U.S.-city, time-series data set to test for an effect of temperature on mortality while adjusting for atmospheric particulate matter with diameter of 2.5 micrometres or less (PM2.5), which is known to be measured with error.
\end{abstract}

\noindent%
{\it Keywords:}  Causal inference, Climate change, Double robustness, Environmental health, Measurement error, Semiparametric inference
\vfill

\newpage
\spacingset{1.45} 

\section{Introduction}\label{sec:intro}
In observational studies across a number of disciplines, it is common to observe variables measured with error. As noted in \cite{cote1987estimating}, ``\cite{campbell1988methodology} has gone so far as to say that measurement error (both random error and method effect) and its confounding influences on research findings cannot be avoided.'' In the field of causal inference, data on covariates are needed to adjust for confounding in order to make inferences with causal interpretations. While a commonly-cited result states that the ordinary least squares (OLS) coefficient estimate of a single variable subject to classical measurement error in a multiple linear regression will merely be attenuated to the null, and hence produce a valid (albeit conservative) statistic of the null hypothesis of no association, the effects of confounders measured with error can be more harmful. Unaccounted for, mismeasured confounders will produce biased effect estimates and invalid hypothesis tests of a treatment effect in even the simplest of settings. Consider a confounder $X^*$ that is measured with classical, nondifferential measurement error $\varepsilon^*$ such that $X=X^*+\varepsilon^*$, where $X$ is the value that is actually observed. A multiple linear regression of outcome $Y$ on exposure $A$ and observed confounders $C$ and $X$ will produce a treatment effect estimate that is biased towards the crude (unadjusted) estimate. Consequently, hypothesis tests concerning the effect of $A$ based on this regression may be invalid in the sense that the corresponding type I error rate will often exceed the nominal level.

In this paper, we present a large class of test statistics of the null hypothesis of no conditional average causal effect that maintain validity when a set of continuous confounders are measured with classical, nondifferential measurement error. We consider three different semiparametric models, all of which specify that the conditional mean of the exposure is linear in the error-prone confounders (on the additive, multiplicative, or logit scale). Beyond this specification, our class of test statistics contains three corresponding subclasses consisting of test statistics with nominal type I error rate within (a) a model that specifies the conditional exposure mean as some possibly-nonlinear function of error-free confounders, (b) a model that specifies the conditional outcome mean as some possibly-nonlinear function of the error-free confounders, and (c) a model that specifies that at least one of (a) or (b) holds. Statistics in (c) are said to be doubly robust.

There has been a great deal of interest in methodology for confounders measured with error. \cite{ogburn2012nondifferential} and \cite{ogburn2013bias} consider settings in which misclassification of a discrete confounder results in imperfect confounding adjustment, in the sense that the average causal effect will be biased in the direction of the crude (unadjusted) estimate. We will focus on settings in which the error-prone confounders are continuous. When instrumental variables (IVs) for such confounders are available, traditional IV estimators can be used to adjust for measurement error in a linear outcome regression model \cite[among others]{amemiya1985instrumental, amemiya1990instrumental, amemiya1990two, buzas1996instrumental, carroll1994measurement, carroll2006measurement, fuller2009measurement, stefanski1995instrumental}. \cite{kuroki2014measurement} give an identification result for a total effect in a linear structural equation model with Gaussian errors when at least two proxies of an error-prone confounder are available. \cite{raykov2012propensity} proposes a propensity-score estimator for the average causal effect under a latent variable model for confounders, in which at least two error-prone congeneric indicators are measured for each latent confounding variable. \cite{mccaffrey2013inverse} propose an inverse probability of treatment weighted (IPTW) estimator with weights that are functions of the error-prone confounders that is consistent for the average causal effect when the distribution of the measurement error is homoscedastic and known or consistently estimated. In addition to IPTW estimators, \cite{lockwood2015matching} also consider matching estimators, establishing necessary and sufficient conditions for recovering unconfounded matching estimators based on functions of error-prone confounders. However, they show that these are less likely to hold in practice than conditions for IPTW, and do not provide general guidance in estimating appropriate functions to match on. \cite{lockwood2015simulation} propose a simulation-extrapolation estimator that assumes normality and known or consistently-estimated variance of the measurement error. \cite{cochran1973controlling} derive an analytical expression characterizing the bias incurred by continuous confounders measured with classical error in a linear regression model. Under simplifying assumptions, this bias can be corrected provided the reliability ratio is known. \cite{battistin2014treatment} generalize this work to nonparametric models, allowing for identification of the average causal effect of treatment and the effect of treatment on the treated in a sensitivity analysis on a range of possible values for the variance of measurement error. 

The latter method fits into a more general body of measurement error research that does not rely on external data. While most traditional measurement-error methods depend on auxiliary data such as instrumental variables or data from reliability or validation studies, attention has more recently shifted to developing methods not dependent on such data, which can be expensive to collect or simply unavailable. One such class of methods uses ``higher-order'' moment restrictions to produce identifying estimating equations for parameters of a regression model with covariates measured with error \citep[among others]{bonhomme2009consistent, cragg1997using, dagenais1997higher, erickson2000measurement, erickson2002two, kapteyn1983identification, lewbel1997constructing, lewbel2012using, pal1980consistent, schennach2013nonparametric, stuart1979advanced}. Another existing method, known as deconvolution, uses external knowledge of the measurement-error distribution to recover the density of the error-free variable \citep[among others]{fan1991optimal, fan1993nonparametric}. Though this idea is attractive in principle, it is very rare that the distribution of measurement error will be known, and convergence rates tend to be too slow for practical use. When identification conditions are not met, it is possible to compute bounds for the parameter of interest \citep{frisch1934statistical, klepper1984consistent, schennach2014entropic}. \cite{carroll2006measurement} and \cite{schennach2012measurement} survey measurement error literature in which they provide a thorough treatment of methods not requiring external data.

This paper contributes both to the literature on confounder measurement error as well as on measurement error methods not requiring external information. Our proposed class of test statistics is of interest in a variety of practical settings in that it requires neither knowledge of the distribution or variance of the measurement error (as in \cite{battistin2014treatment} or in deconvolution), nor any form of external information. We will assume that the measurement error is mean independent of the error-free confounders and the outcome, i.e. $E(\varepsilon^*\mid C,Y)=E(\varepsilon^*)$ over the joint support of $C$ and $Y$. Otherwise, no other moment restriction is required not already embedded in the assumptions needed to draw causal inferences. In particular, our statistics directly leverage the no-unobserved-confounding assumption needed for identification of the average causal effect, even in the absence of confounder measurement error.

The governing idea of the proposed approach is that under the null hypothesis of no effect of exposure, the assumption of no unobserved confounding renders the outcome an instrumental variable for the association between the true error-prone covariate $X^*$ and $A$ adjusting for $C$. Thus, as documented in the literature on IV methods for measurement error, $Y$ can be used to obtain a consistent estimator of the association between $(C,X^*)$ and $A$ \cite[among others]{amemiya1985instrumental, amemiya1990instrumental, amemiya1990two, buzas1996instrumental, carroll1994measurement, carroll2006measurement, fuller2009measurement, stefanski1995instrumental}. In this paper, we show that estimation of the conditional association between error-prone covariates and exposure can be accomplished jointly with a test of no treatment effect under a unifying framework of a generalized method of moments test based on overidentifying moment restrictions, known in the econometrics literature as a Sargan test \citep{sargan1958estimation}, Hansen test, or J-test \citep{hansen1982large}.

We demonstrate validity of our test statistics in the presence of measurement error in an extensive simulation study, and compare them with standard outcome-regression and g-estimation tests that do not allow for measurement error. In simulation settings, our tests retain validity while the competing tests break down in the presence of measurement error. We also apply our methods to an environmental health data set to test for a causal effect of same-day temperature on mortality in the United States. We conduct a multi-city analysis with daily information on mortality as well as environmental factors including temperature and concentration of particulate matter with diameter of 2.5 micrometres or less (PM2.5). PM2.5 is known to be a confounder and to be measured with error due to the high level of variability of pollution across monitoring stations \citep{armstrong1990effects, armstrong2004exposure, bateson2007panel, kioumourtzoglou2014exposure, zeger2000exposure}. Temperature is hypothesized to have a causal association with mortality, allowing us to examine our method to test this hypothesis while being robust to confounder measurement error.

\section{A Class of Propensity-Score-Based Test Statistics Robust to Measurement Error}\label{sec:class}

To formalize discussion, we define for each $a$ the counterfactual $Y_a$ to be a subject's outcome had the subject been assigned, possibly contrary to fact, to exposure level $a$. We link these counterfactuals to the observed variables via the consistency assumption \citep{robins1986new}, which states that if $A=a$, then $Y_a=Y$ with probability one for each level $a$. Suppose we observe a set of covariates $C$ that are measured without error as well as an additional set of covariates $X$ that are measured with classical error, i.e., additive measurement error. The latter are related to their corresponding, unobserved, true value $X^*$ by $X=X^*+\varepsilon^*$, where $\varepsilon^*$ is the measurement error, assumed to be mean independent of $C$ and $Y$. Further, suppose that given $X^*$ and $C$, there is no unmeasured confounding of the effect of $A$ on $Y$, which can be formalized as follows:
\begin{Assumption}\label{assumption}
$Y_a\ci A\mid C,X^*$ for each level $a$ (No unmeasured confounding).
\end{Assumption}
\noindent Assume that A is continuous; results are generalized to binary and count exposure in Section \ref{sec:exten}. We now present a class of test statistics for the null hypothesis
$H_0: E(Y_a\mid C,X^*)=E(Y_0\mid C,X^*)$ for all $a$.
Intuitively, under the stronger sharp null, $Y_a=Y$ w.p. 1 for all $a$, the assumption of no unmeasured confounding implies that $A\ci Y\mid C,X^*$. Furthermore since $X^*$ is a confounder, we have that $Y\nci X^*\mid C$. These two statements formally define $Y$ as an instrumental variable for the conditional association between $X^*$ and $A$ given $C$, and therefore can be used to account for measurement error in estimating a model of $A$ given $C$ and $X$ \citep{carroll2006measurement}. Although H$_0$ is technically weaker than the sharp null, as we will show the essential idea that $Y$ can nonetheless be used under the null to correct for measurement error in the exposure model remains true, despite $Y$ no longer formally being an IV.

We present our first result, which relies on correct specification of a mean regression model for exposure. We will refer to this as the propensity-score model \citep{rosenbaum1983central}. In this vein, consider the semiparametric model  $\mathcal{M}_A$ as the set of laws for $(A,X^*,C,Y)$ with sole restriction the parametric model $E(A\mid C,X^*;\alpha)=[1,X^{*T}]g_{A}(C;\alpha)$, where $g_{A}=[g_{A,1}(C;\alpha_1),\allowbreak g_{A,2}(C;\alpha_2)^T]^T$ is a known function of $C$ indexed by the unknown parameter $\alpha=[\alpha_1^T,\alpha_2^T]^T$, where $g_{A,1}$ is real valued and $g_{A,2}$ has the same dimension as $X^*$. We define $p_1$ and $p_2$ to be the dimensions of $\alpha_1$ and $\alpha_2$, respectively, such that $\alpha$ has dimension $p\equiv p_1+p_2$. We assume throughout that the conditional $X^*-A$ association is linear given $C$, however $g_{A}(C;\alpha)$, though parametric, can be a nonlinear function of $C$. The case $g_{A,2}(C;\alpha_2)=\alpha_2$ is a constant implies no $X^*-C$ interaction in the model for $A$. Define $\nabla_\alpha$ to be the gradient operator with respect to $\alpha$ and $\mathbb{P}_n$ to be the empirical mean operator.

\begin{Theorem}\label{thm:PS}
Let $\ell(C)$ and $m(C)$ be $p+q$-dimensional functions of $C$ for some positive integer $q$, such that the elements of $\ell(C)Y+m(C)$ are linearly independent. Define
$U(\alpha)\equiv\left\{\ell(C)Y+m(C)\right\}\{A-[1,X^T]g_{A}(C;\alpha)\}$, $\Omega\equiv E\left\{U(\alpha)U(\alpha)^T\right\}$ and $\hat{U}_n(\alpha)\equiv \mathbb{P}_n U_i(\alpha)$. If $U(\alpha)$ is continuously differentiable, $\nabla_\alpha E\{U(\alpha)\}=E\{\nabla_\alpha U(\alpha)\}$, and $\Omega^{-1} E\{\nabla_\alpha U(\alpha)\}$ has full rank, then for any $\hat{\Omega}_n\xrightarrow{p}\Omega$, the test statistic $\chi^2_{rps}\equiv\min\limits_{\alpha}n\hat{U}_n(\alpha)^T\hat{\Omega}_n^{-1}\hat{U}_n(\alpha)\xrightarrow{d} \chi^2_q$
under $\mathcal{M}_A$ and H$_0$.
\end{Theorem}

Thus, we have a valid test of no causal effect of treatment which depends on $X^*$ only through the mismeasured covariate, $X$. Intuitively, standard normal equations for the propensity-score model incur bias due to components that include the product of the residual $A-E(A\mid C,X)$ with the error-prone covariate $X$. However, this can be amended by replacing one of these components with the product of the residual with $Y$. Under H$_0$, this product will form an unbiased estimating equation. Additional unbiased estimating functions can be added simply by multiplying this latter product with any function of $C$, and hence these can be used to form a valid Sargan test statistic. Thus, a simple form of the test in Theorem \ref{thm:PS} with $q=1$ could use $\nabla_{\alpha}g_A(C;\alpha)$ augmented by $Y$ and the product of $Y$ with an element of $C$ in place of $\ell(C)Y+m(C)$, for instance. In order to ensure linear independence of the elements of $\ell(C)Y+m(C)$, the interaction function $g_{A,2}(C,\alpha_2)$ in the propensity-score model cannot be saturated in $C$.

The following iterative procedure can be used to compute the variance-estimate component $\hat{\Omega}_n$:
\begin{algorithm2e}
initialize $\tilde{\alpha}:=\argmin\limits_\alpha\hat{U}_n(\alpha)^T\hat{U}_n(\alpha)$\;
set $\hat{\Omega}_n:=\mathbb{P}_n\left\{U(\tilde{\alpha})U(\tilde{\alpha})^T\right\}$\;
\Do{convergence not reached}{
	set $\tilde{\alpha}:=\argmin\limits_\alpha \hat{U}_n(\alpha)^T\hat{\Omega}_n^{-1}\hat{U}_n(\alpha)$\;
	set $\hat{\Omega}_n:=\mathbb{P}_n\left\{U(\tilde{\alpha})U(\tilde{\alpha})^T\right\}$\;
}
\end{algorithm2e}

The first two steps are in fact sufficient for asymptotic validity, however iterating generally improves finite-sample performance. Alternatively, a continuous updating approach can be used, in which $\hat{\Omega}_n$ is indexed by $\alpha$, and $n\hat{U}_n(\alpha)^T\hat{\Omega}_n(\alpha)^{-1}\hat{U}_n(\alpha)$ is minimized in $\alpha$ through both $\hat{U}_n(\alpha)$ and $\hat{\Omega}_n(\alpha)$.

\section{A Class of Doubly-Robust Test Statistics}\label{sec:dr}

Validity of the test statistic given in the previous section relies on correct specification of $E(A\mid X^*,C)=[1,X^*]^Tg_A(C;\alpha)$, however this model may be misspecified. Therefore it is of interest to explore an alternative, potentially more robust approach. Here we present a large class of doubly-robust test statistics. In order to describe this class, let $E(Y\mid C;\gamma)=g_Y(C;\gamma)$ denote a parametric model for $E(Y\mid C)$, and consider the semiparametric model $\mathcal{M}_Y$ with sole restrictions 
$E(A\mid C,X^*)-E(A\mid C,X^*=0)=X^{*T}g_{A,2}(C;\alpha_2)$ and
$E(Y\mid C)=g_Y(C;\gamma)$. This is a semiparametric model since the association between $C$ and $A$ given $X^*=0$ is unrestricted. Further consider the union model $\mathcal{M}_\cup\equiv\mathcal{M}_A \cup \mathcal{M}_Y$. We present a class of test statistics for each of these two models, adopting the notation $\Delta_A(\alpha)\equiv A-[1,X^{T}]g_{A}(C;\alpha)$ and $\Delta_Y(\gamma)\equiv Y-g_Y(C;\gamma)$ for the residuals in each model.

\begin{Theorem}\label{thm:OR}
Let
\[U(\alpha_2,\gamma)\equiv\left[\begin{array}{c}
k(C)\Delta_Y(\gamma)\{A-X^{T}g_{A,2}(C;\alpha_2)\}\\
S(\gamma)
\end{array}\right],\]
where $S(\gamma)$ is a system of estimating equations for $\gamma$, and $k(C)$ is a vector-valued function of $C$ with linearly-independent elements with dimension $p_2+q$ for some positive integer $q$. Suppose $U(\alpha_2,\gamma)$ is continuously differentiable, $\nabla_{\alpha_2,\gamma} E\{U(\alpha_2,\gamma)\}=E\{\nabla_{\alpha_2,\gamma} U(\alpha_2,\gamma)\}$, and $\Omega^{-1}\allowbreak E\{\nabla_{\alpha_2,\gamma} U(\alpha_2,\gamma)\}$ has full rank, where $\Omega=\allowbreak E\{\allowbreak U(\allowbreak\alpha_2,\allowbreak\gamma)\allowbreak U(\allowbreak \alpha_2,\allowbreak \gamma)^T\}$. Then under $\mathcal{M}_Y$ and H$_0$, $\chi^2_{ror}\equiv\min\limits_{\alpha_2,\gamma}n\hat{U}_n(\alpha_2,\allowbreak\gamma)^T\allowbreak\hat{\Omega}_n^{-1}\hat{U}_n(\alpha_2,\gamma)\xrightarrow{d} \chi^2_q$ for any $\hat{\Omega}_n\xrightarrow{p}\Omega$.
\end{Theorem}
We also have the result:
\begin{Theorem}\label{thm:DR}
Let
\[U(\alpha,\gamma)\equiv\left[\begin{array}{c}
k(C)\Delta_Y(\gamma)\Delta_A(\alpha)\\
\left\{\ell(C)Y+m(C)\right\}\Delta_A(\alpha)\\
S(\gamma)
\end{array}\right],\]
where $S(\gamma)$ is a system of estimating equations for $\gamma$ that is unbiased when $g_Y(C;\gamma)$ is correctly specified, $k(C)$, $\ell(C)$, and $m(C)$ are each vector-valued functions of $C$ such that $k(C)$ and $\ell(C)Y+m(C)$ each consist of linearly-independent elements, and $\ell$ and $m$ have dimension $p_1$ and $k$ has dimension $p_2+q$ for some positive integer $q$. Suppose $U(\alpha,\gamma)$ is continuously differentiable, $\nabla_{\alpha,\gamma} E\{U(\alpha,\gamma)\}=E\{\nabla_{\alpha,\gamma} U(\alpha,\gamma)\}$, and $\Omega^{-1} E\{\nabla_{\alpha,\gamma} U(\alpha,\gamma)\}$ has full rank, where $\Omega=E\left\{U(\alpha,\gamma)U(\alpha,\gamma)^T\right\}$. Then under $\mathcal{M}_\cup$ and H$_0$, $\chi^2_{dr}\equiv\min\limits_{\alpha,\gamma}n\hat{U}_n(\alpha,\gamma)^T\hat{\Omega}_n^{-1}\hat{U}_n(\alpha,\gamma)\xrightarrow{d} \chi^2_q$, for any $\hat{\Omega}_n\xrightarrow{p}\Omega$.
\end{Theorem}

As before, an appropriate variance estimator $\hat{\Omega}_n$ can be computed using either an iterated procedure or a continuous-updating approach. An alternative approach would be to first estimate $\gamma$ by solving $\mathbb{P}_n S(\gamma)=0$, plug this value into $U(\alpha,\gamma)$ (rendering the $\gamma$-estimating-equation component zero), and use
\begin{align*}
\hat{\Omega}_n = \frac{1}{n}\sum\limits_{i=1}^n&\left[U_i(\hat{\alpha},\hat{\gamma})-\left\{\sum\limits_{j=1}^n\nabla_\gamma U_j(\hat{\alpha},\gamma)\mid_{\hat{\gamma}}\right\}\left\{\sum\limits_{j=1}^n\nabla_\gamma S_j(\hat{\gamma})\right\}^{-1}S_i(\hat{\gamma})\right]^{\otimes 2}
\end{align*}
for the variance estimator in the denominator of the test statistic. All estimates of $\hat{\Omega}_n$ discussed here require that $k$, $\ell$, and $m$ have not been estimated. Power for both these and the previous tests can be optimized by using appropriate choices of the functions $\ell(C)$, $m(C)$, and $k(C)$ based on the direction of the alternative hypothesis, which we discuss further in Section \ref{sec:power}.

\section{A Simulation Study Demonstrating Validity}\label{sec:sim_dr}

We now present results from a simulation study drawing samples from the following data generating mechanism. We generate $(Y_0,C)$ under a joint normal model given by $Y_0 = N(0,1)$ and $C = Y_0+N(0,1)$, and $X^*$ and $A$ under $X^* = Y_0+C+Y_0C+N(0,1)$ and $A = C+X^*+N(0,4)$. To reflect the null hypothesis, we let $Y = Y_0$. We generate $X$ from the classical measurement error model $X = X^*+N\{0,9(1/\tau-1)\}$, where $\tau$ is the reliability ratio, i.e., the ratio of the variability of the true variable $X^*$ to the variable measured with error $X$. One may easily verify that Assumption \ref{assumption} and H$_0$ are satisfied.

We drew 100,000 samples of size 5000 under four settings with reliability ratios of 50\%, 70\%, 90\%, and 100\% (i.e., no measurement error). In each setting, we applied the three testing procedures given in Sections \ref{sec:class} and \ref{sec:dr}. We compared these tests with two others that ignored the presence of measurement error. The first was an outcome-regression-based test, using the p-value of the regression coefficient for $A$ when regressing $Y$ on $C$, $X$, and $A$ using OLS and using a sandwich variance estimate. Though the outcome model is not correctly specified, the OLS estimate of the coefficient for $A$ in the absence of measurement error will be unbiased for the slope of $A$ in $E(Y\mid A,C,X^*)$ (zero). This is because it is equal to the OLS estimate of the coefficient for the residual obtained from a linear regression of $A$ on $C$ and $X^*$, which is correctly specified. The second comparison test was based on g-estimation \citep{robins1989analysis}, using the p-value of the regression coefficient for $Y$ when regressing $A$ on $C$, $X$, and $Y$. All tests used an $\alpha$ level of 0.05. 

All tests with the exception of the standard outcome-regression test were conducted under three different models: the intersection model $\mathcal{M}_{\cap}$, in which both $g_Y(C;\gamma)$ and $g_A(C;\alpha)$ were correctly specified; $\mathcal{M}_Y$, in which $g_Y(C;\gamma)$ was correctly specified and $g_{A,1}(C;\alpha_1)$ was not; and $\mathcal{M}_A$, in which $g_A(C;\alpha)$ was correctly specified and $g_Y(C;\gamma)$ was not. We used $g_Y(C; \gamma) = \gamma_0+\gamma_1C$ and $g_{A,1}(C; \alpha_1) = [1,C]\alpha_1$ for correctly-specified models and $g_Y(C; \gamma) = \gamma_0+\gamma_1C^2$ and $g_{A,1}(C; \alpha_1) = [1,C^2]\alpha_1$ for incorrectly-specified models. Though the conditional mean of $Y$ given $C$ and $A$ does not have a simple form, the standard outcome-regression test does not require it to be modeled correctly for validity. Therefore, we show results for the standard outcome-regression test using a misspecified model in all cases for the purposes of comparison. The index functions for the doubly-robust test used an orthonormalization of $[1,C,C^2,C^3]^T$, with $k(C)$ and $m(C)$ being equal to the first two rows and $\ell(C)$ being equal to the last two. For the robust propensity-score test, we used a Gram-Schmidt orthonormalization of $[1,C,C^2,C^3]^T$ for the function $\ell(C)$, $m(C)=0$, and $q=1$. For the robust outcome-regression test, we used $k(C)=[1,C]^T$. The score equations for $\gamma$ in the doubly-robust and robust outcome-regression tests were $S(\gamma)=[1,C^2]^T(Y-\gamma_0-\gamma_1C^2)$ under $\mathcal{M}_A$, and $S(\gamma)=[1,C]^T(Y-\gamma_0-\gamma_1C)$ otherwise. Results are presented in Table \ref{tab:sim}.

\begin{table}
\centering
\caption{Estimated type 1 error from 100,000 hypothesis tests simulated under the null hypothesis\label{tab:sim}}
\begin{tabular}{c c r@{.}l r@{.}l r@{.}l r@{.}l r@{.}l}
\\
\hline\hline
Model & Rel. ratio (\%) &\multicolumn{2}{c}{DR}&\multicolumn{2}{c}{Robust PS}&\multicolumn{2}{c}{Robust OR}&\multicolumn{2}{c}{G-estimation}&\multicolumn{2}{c}{Standard OR} \\
\hline
$\mathcal{M}_{\cap}$	& 50	& 0&0455	& 0&0472	& 0&0485 & \multicolumn{2}{l}{1}	& \multicolumn{2}{l}{1} \\
&	70	& 0&0453	& 0&0482	& 0&0497 & \multicolumn{2}{l}{1}	& 1&000 \\
&	90	& 0&0533	& 0&0484	& 0&0505 & 0&645	& 0&643 \\
&	100	& 0&0476	& 0&0489	& 0&0517 & 0&0496	& 0&0493 \\
$\mathcal{M}_Y$	& 50	& 0&0452	& 0&714 & 0&0485	& \multicolumn{2}{l}{1}	& \multicolumn{2}{l}{1} \\
&	70	& 0&0486	& 0&858	& 0&0497 & \multicolumn{2}{l}{1}	& 1&000 \\
&	90	& 0&0527	& 0&939	& 0&0505 & \multicolumn{2}{l}{1}	& 0&643 \\
&	100	& 0&0519	& 0&958	& 0&0517 & 0&965	& 0&0493 \\
$\mathcal{M}_A$	& 50	& 0&0463	& 0&0472	& \multicolumn{2}{l}{1} & \multicolumn{2}{l}{1}	& \multicolumn{2}{l}{1} \\
&	70	& 0&0495	& 0&0482 & \multicolumn{2}{l}{1}	& \multicolumn{2}{l}{1}	& 1&000 \\
&	90	& 0&0472	& 0&0484	& \multicolumn{2}{l}{1} & 0&645	& 0&643 \\
&	100	& 0&0497	& 0&0489	& \multicolumn{2}{l}{1} & 0&0496	& 0&0493 \\
\hline
\\
\end{tabular}
\end{table}

As expected, the doubly-robust test was approximately valid with correct Monte Carlo type 1 error rate under all settings. The robust propensity-score test and robust outcome regression test, on the other hand, were approximately valid under all settings apart from under $\mathcal{M}_Y$ and $\mathcal{M}_A$, respectively. G-estimation and standard outcome-regression tests were not valid in the presence of measurement error, and the standard outcome-regression test was approximately valid in its absence. The g-estimation test was approximately valid under no measurement error only when $g_A$ was correctly specified.

\section{Application to Test for an Effect of Temperature on Mortality}\label{sec:data}

As evidence for climate change continues to accumulate, the natural question of whether temperature affects mortality is of increasing importance. While there are many long-term threats posed by rising global temperatures, the immediate effects on mortality also pose a grave public-health concern. When studying this effect, it is vital to control for air pollution as a potential confounder \citep{o2003modifiers}. A common metric of air pollution is PM2.5 concentration, however this is well known to be measured with error \citep{armstrong1990effects, armstrong2004exposure, bateson2007panel, kioumourtzoglou2014exposure, zeger2000exposure}. In particular, PM2.5 is considered to be contaminated with a mixture of both Berkson error, due to the variability of concentration actually experienced across individuals, and classical error, due to aggregation of measurements across multiple monitoring stations \citep{kioumourtzoglou2014exposure, zeger2000exposure}. The former is benign in the sense that it increases variance but introduces no bias; it is the latter with which we are most concerned. Some studies try to reduce measurement error by using spatial smoothing models \citep{hoek2002association, jerrett2005spatial, puett2009chronic, sampson2011pragmatic, szpiro2010predicting, yanosky2008predicting}, however these rely on geographical data on residency and may induce other forms of error \citep{gryparis2009measurement, sheppard2012confounding, szpiro2011efficient}. We implemented our method, and compared it against two methods that ignore the presence of measurement error.

The data set used here consists of time-series mortality data from forty-one U.S. cities measured over the course of 1999 to 2006, though in our analysis, we only considered twenty-four cities with at least eight deaths per day, as cities with lower mortality rates were unlikely to provide enough power to detect an effect. Data on individual mortality with exact date of death was acquired from the National Center for Health Statistics (NCHS) and from state public health departments \citep{zanobetti2009effect}. We excluded accidental deaths (ICD-code 10th revision: V01-Y98, ICD-code 9th revision: 1-799) and deaths of individuals who did not reside in the city in which they died. Temperature data were obtained from the National Oceanic and Atmospheric Administration (NOAA) website, with a city being assigned ambient temperature readings from its nearest monitoring station. PM2.5 data were obtained from the US Environmental Protection Agency's (EPA) Air Quality System (AQS) database (US EPA 2013). PM2.5 readings were averaged over all monitors in a city whenever multiple readings were available.

On a given day, $i$, let $Y_i$ denote the number of deaths, $A_i$ denote the average temperature in degrees Celsius, $X_i$ denote the average PM2.5 concentration measurement, and $C_i$ consist of date, $t_i$, and dummy variables for day of week. The functions $g_{A,1}(C;\alpha_1)$ and $g_Y(C;\gamma)$ in the propensity-score and outcome-regression models used both Fourier bases for time with a period of one year to account for seasonal trends as well as polynomial bases for time to account for secular trends. The dimensions of the Fourier bases were at least four (not including intercept) and the dimensions of the polynomial bases started at zero. Sargan goodness-of-fit tests were used to assess model fit, and more dimensions were added to the bases until the tests no longer rejected at an $\alpha$ level of 0.10. In particular, we used forms of the test in Theorem \ref{thm:PS} with $\ell(C)=0$ (eliminating its power to test for an effect on $Y$) and $m(C)$ equal to $\nabla_\alpha g_A(C;\alpha)$ augmented by the next two polynomial or Fourier basis functions. Analogous tests were used for the robust outcome-regression model, with the moment functions being equal to the regression residuals multiplied by $\nabla_\gamma g_Y(C;\gamma)$ augmented by the next two polynomial or Fourier basis functions. Without this step, test rejections could be attributable to model misspecification rather than the presence of a true effect. Both models were linear in these terms as well as the day-of-week dummy variables, and the propensity-score model did not include interaction between $C$ and $X$. The outcome-regression model used a log link.

Measurement-error-robust tests from each of the three classes presented in Theorems \ref{thm:PS}-\ref{thm:DR} were conducted based on these models. For the robust propensity-score test, we used $[1, t, 0_{p-1}^T]^T$ as the function $\ell(C)$, where $0_{p-1}$ is a vector of zeroes with length $p-1$, and $[0,0,\nabla_{\alpha_1}^T g_{A,1}(C;\alpha_1)]^T$ as the function $m(C)$. For the robust outcome-regression test, we used $[1,t^r]^T$ as the function $k(C)$, where $r$ is the smallest order of polynomial not included in $g_{A,1}(C;\alpha_1)$. For the doubly-robust test, we used $[1,t]^T$ as the function $k(C)$, $\nabla_{\alpha_1} g_{A,1}(C;\alpha_1)$ as $m(C)$, and a vector of zeroes with length $p_1$ as the function $\ell(C)$. Thus, in each case $q=1$, and we compared resulting test statistics with the corresponding null distribution, $\chi^2_1$. We used the doubly-robust test for inference, and supplemented our analysis with the other two for an additional check of model fit. As previously mentioned, no model of the relationship between temperature and mortality is needed. This is particularly advantageous in our setting, since this relationship tends to be V or J shaped, and hence not as simple to model.

The two standard statistics considered in the simulation study -- based on the g-estimation and standard outcome-regression tests -- were implemented. These standard tests do require models for the relationship between temperature and mortality. Consequently, for the standard outcome-regression test, we used a quasi-Poisson model for the outcome-regression model with the deterministic component consisting of the same $g_Y(C;\gamma)$ described above for the robust test statistics, plus terms for linear-spline basis functions of temperature. We placed a single knot for this spline at 16 $\degree$C, which is around where vertices of this nonmonotone relationship tend to be. Due to the non-invertibility of this relationship, we conducted two standard g-estimation tests: one based on data from days with mean temperature no higher than 16 $\degree$C, and the other from days with mean temperature no lower than 16 $\degree$C. For both, we used OLS estimation of the regression of temperature on the same $g_A(C;\alpha)$ described above for the robust test statistics plus a linear term for the outcome. Sargan tests were implemented using the gmm package in R. To account for the serially-correlated nature of our data, we used heteroskedasticity and autocorrelation consistent variance estimators from the sandwich package in R for all tests. Results are presented in Table \ref{tab:data}.

\begin{table}
\centering
\caption{P-values of hypothesis tests for an effect of temperature on mortality in U.S. cities\label{tab:data}}
\begin{tabular}{l r@{.}l r@{.}l r@{.}l r@{.}l r@{.}l r@{.}l}
\\
\hline\hline
 & \multicolumn{2}{c}{Robust} & \multicolumn{2}{c}{Robust} & \multicolumn{2}{c}{Doubly} & \multicolumn{2}{c}{G-est.} & \multicolumn{2}{c}{G-est.} & \multicolumn{2}{c}{Standard} \\
City & \multicolumn{2}{c}{PS} & \multicolumn{2}{c}{OR} & \multicolumn{2}{c}{Robust} & \multicolumn{2}{c}{$\leq$16 \degree C} & \multicolumn{2}{c}{$\geq$16 \degree C} & \multicolumn{2}{c}{OR} \\
\hline
Albuquerque, NM			&	0&18		&	0&22		&	0&20		&	0&33	&	0&79		&	0&56 \\
Allentown, PA				&	0&016		&	0&15		&	0&92		&	0&88	&	0&58		&	0&25 \\
Annandale, VA				&	0&52		&	0&070		&	0&84		&	0&91	&	0&39		&	0&98 \\
Baltimore, MD 				&	0&94 		& 0&29 		&	0&82 		&	0&30	&	0&20		&	0&0018 \\
Boston, MA 					&	0&015 		& 0&083 		&	0&50		&	0&47	&	0&055		&	0&00073 \\
Elizabeth, NJ					& 0&27 		& 0&14 		&	0&61		&	0&69	&	0&091		&	0&052 \\
Hartford, CT 				&	0&19 		& 0&20 		&	0&13		&	0&36	&	0&22		&	0&033 \\
Lancaster, PA 				&	0&17 		& 0&067 		&	0&16		&	0&21	&	0&16		&	0&54 \\
Melville, NY 					&	0&65 		& 0&47 		&	0&25		&	0&94	&	0&045		&	0&0504 \\
Middlesex, NJ 				&	0&48 		& 0&32 		&	0&18 		&	0&91	&	0&44		&	0&38 \\
New Haven, CT 			&	0&65 		& 0&067 		&	0&79 		&	0&73	&	0&016		&	0&0047 \\
New York, NY 				&	0&013 		& 6&1e-4 	&	0&0018	&	0&86	&	0&0010	&	6&7e-10 \\
Newark, NJ 					&	0&58 		& 0&0085 	&	0&48 		&	0&18	&	0&36		&	0&039 \\
Paterson, NJ 				&	0&66 		& 0&13 		&	0&33		&	0&47	&	0&11		&	0&64 \\
Philadelphia, PA 			&	0&0083	& 0&014		&	0&20		&	0&29	&	0&092		&	7&9e-9 \\
Reading, PA 					&	0&64 		& 0&019 		&	0&69		&	0&45	&	0&79		&	0&49 \\
Richmond, VA 				&	0&25 		& 0&89 		&	0&34		&	0&19	&	0&91		&	0&60 \\
Salt Lake City, UT			&	0&56 		& 8&8e-4 	&	0&77		&	0&014	&	0&10		&	0&0033 \\
Spokane, WA 				&	0&55 		& 0&95 		&	0&75		&	0&082	&	0&011		&	0&011 \\
Stamford, CT 				&	0&060 		& 0&37 		&	0&034		&	0&95	&	0&23		&	0&54 \\
Upper Marlboro, MD	&	0&11 		& 0&046 		&	0&84		&	0&91	&	0&10		&	0&51 \\
Washington, DC 			&	0&042 		& 0&013	 	&	0&80		&	0&20	&	0&97		&	0&0023 \\
Wilmington, DE 			&	0&64 		& 0&13 		&	0&87		&	0&59	&	0&51		&	0&66 \\
York, PA 						&	0&27 		& 0&17 		&	0&28		&	0&35	&	0&92		&	0&68 \\
\hline
\\
\end{tabular}
\end{table}

The doubly-robust test rejected the null hypothesis of no effect of temperature on mortality in two cities: New York, NY and Stamford, CT. For New York, the other measurement-error-robust tests also rejected, whereas these tests did not reject for Stamford. In the latter case, this suggests that the other measurement-error-robust tests were underpowered relative to the doubly-robust test. For New York, there was no indication of substantial attenuation due to measurement error, as all standard tests rejected, apart from the g-estimation test on colder days. For Stamford, on the other hand, it does appear that measurement error may have masked an effect from the standard tests, as all standard tests failed to reject.

The robust outcome-regression test rejected for Newark, Reading, Salt Lake City, and Upper Marlboro, while neither the doubly-robust test nor the robust propensity-score test did, suggesting that either the outcome-regression model may not be correctly specified for these cities, or that the other measurement-error-robust tests were relatively underpowered. Similarly, the robust propensity-score test rejected for Allentown and Boston, while neither the doubly-robust nor the robust outcome-regression test did, suggesting that either the propensity-score model may not be correctly specified, or that the other measurement-error-robust tests were relatively underpowered. Both the robust propensity-score and robust outcome-regression tests rejected for Philadelphia and Washington, DC, while the doubly-robust test did not. This may reflect the fact that the doubly-robust test had less power to reject than the other tests in these cities. In a number of cities, one or more of the standard tests rejected when the robust tests did not. Because of possible bias induced by measurement error, we cannot discern whether these rejections are indicative of true effects or merely artifacts.

\section{Power and Estimation Under an Additive Causal Model}\label{sec:power}

We now consider estimation and testing under the alternative hypothesis of an additive causal model,
\begin{align}\label{eq:causalmodel}
E(Y_0\mid A,C,X^*) = E(Y-\psi_0 A\mid A,C,X^*),
\end{align}
where $\psi_0$ is the average causal effect for a unit change in $A$ such that $\psi_0 A=E(Y_A-Y_0\mid C,X^*)$, and $\psi_0$ is the causal parameter of interest. The following discussion can be easily adapted to other structural mean models, however.

We conducted a supplementary simulation study varying the value of $\psi_0$ to demonstrate the local power of our proposed test statistics, using the same data generating mechanism as in Section \ref{sec:sim_dr} in each of the same measurement error settings, but with $Y=Y_0+\psi_0 A$ to encode the alternative hypothesis. The robust propensity-score test was conducted on 1000 samples of size 5000 for each value of $\psi_0$. Results are presented in Figure \ref{fig:5000}.
\begin{figure}
\centering
\includegraphics[scale=.75]{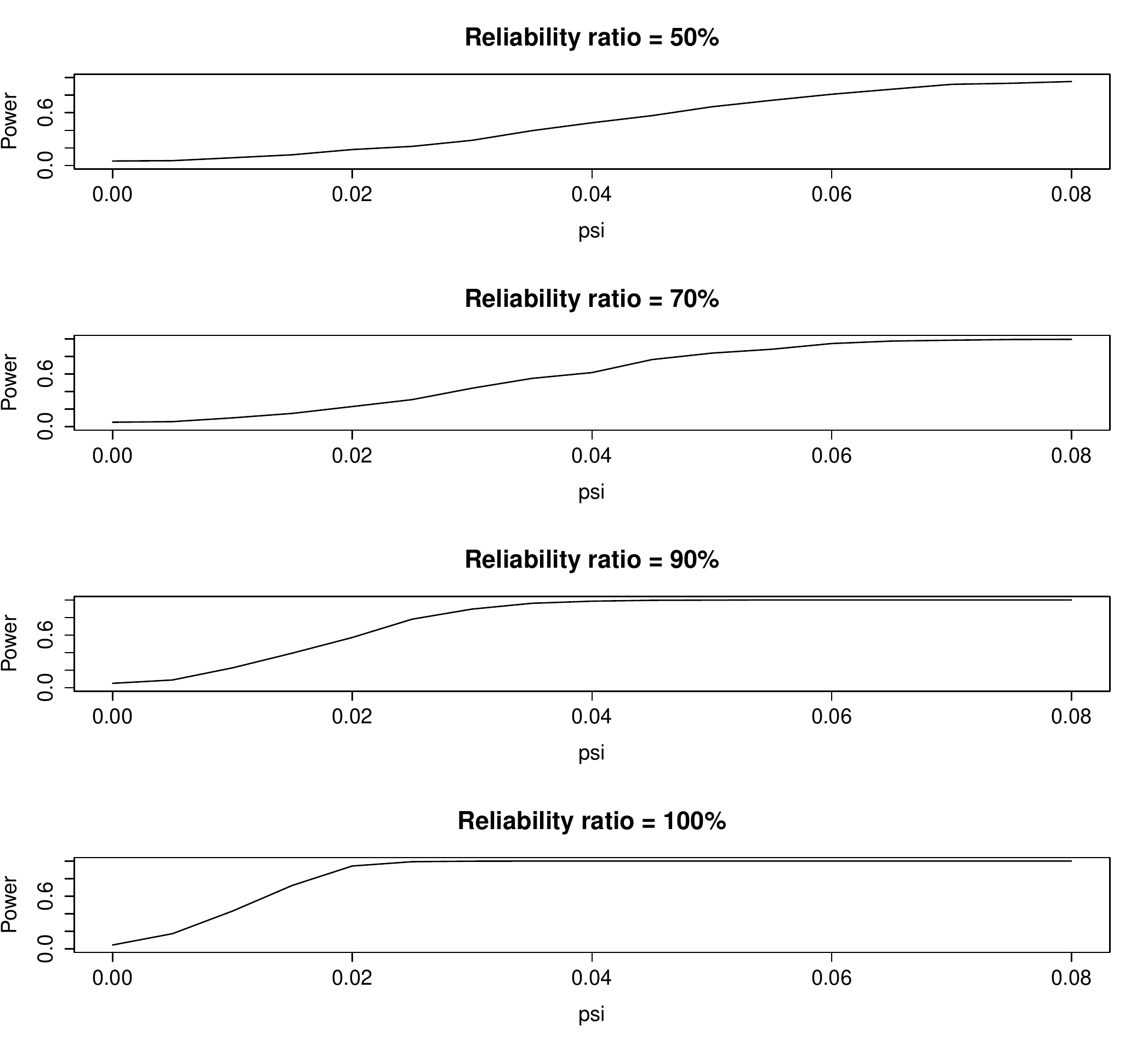}
\caption{Simulation results demonstrating power for n=5000.\label{fig:5000}}
\end{figure}
As expected, we observed trends of increasing power with effect size and reliability ratio. Also as expected, power was approximately 0.05 for $\psi_0=0$ in all cases. The test achieved an estimated 80\% power at $\psi_0=0.06$ when $\tau=0.5$, at $\psi_0=0.05$ when $\tau=0.7$, at $\psi_0=0.025$ when $\tau=0.9$, and at $\psi_0=0.02$ when $\tau=1$. Power appeared to be tending towards unity as $\psi$ increased in all cases. Similar trends were observed in studies with sample sizes of 1000 and 10,000.

Having posited a model for the effect of $A$ on $Y$, our testing approach can be extended for effect estimation. Let $H(\psi) \equiv Y-\psi A$, and define $H \equiv Y-\psi_0 A$ so that $H=H(\psi_0)$. Then $E(H\mid A,C,X^*) = E(Y_0\mid A,C,X^*)$. Our claim is now that $H$ (instead of $Y$, since H$_0$ is no longer assumed) behaves like an instrumental variable for the $X^*$ -- $A$ association, controlling for $C$. To see this, first note that by randomization of $A$ within $\{C,X^*\}$, we have $E(H\mid A,C,X^*)=E(Y_0\mid A,C,X^*)=E(Y_0\mid C,X^*)$, hence $E(H\mid A,C,X^*)=E(H\mid C,X^*)$. Secondly, $\varepsilon^*$ is mean independent of $H$ by assumption. Finally, since $X^*$ is a confounder of the $A$--$Y$ association, $X^*$ must be correlated with $Y_0$, and hence $H$, by definition.

Replacing $Y$ with $H(\psi)$ in the equations given in Theorems \ref{thm:PS}-\ref{thm:DR} when $q=\dim(\psi)$ (one, under model \ref{eq:causalmodel}) produces a system of estimating equations for $\psi$, $\alpha$, and (in the doubly-robust and outcome-regression cases) $\gamma$. Unbiasedness of these functions follows analogously to unbiasedness of the moment equations shown in the proofs for Theorems \ref{thm:PS}-\ref{thm:DR} in the Appendix. Thus, unknown parameters can be estimated by solving $\mathbb{P}_nU(\psi,\alpha)=0$, $\mathbb{P}_nU(\psi,\alpha_2,\gamma)=0$, or $\mathbb{P}_nU(\psi,\alpha,\gamma)=0$. Under certain regularity conditions, the resulting estimator will be consistent and asymptotically normal, and the estimator produced by solving the doubly-robust estimating equations will have these properties provided at least one of $g_A(C;\alpha)$ or $g_Y(C;\gamma)$ is specified correctly. However, even in simple linear models, the ``profile estimating equations'' in $\psi$ can be highly nonlinear. By profile estimating equations, we mean the equations obtained by solving $\dim(\alpha)+\dim(\gamma)$ of the estimating equations for $\alpha$ and $\gamma$ implicitly in terms of $\psi$, and plugging these into the additional estimating equations. Identifiability issues can be mitigated by using overidentified estimating equations, i.e., $q>\dim(\psi)$, but solving these may remain challenging, and can result in very unstable estimation.

Optimal choices of functions $\ell(C)$ and $m(C)$ for the robust propensity-score estimator are given in the Appendix. These functions also optimize power when used for hypothesis testing as in Section \ref{sec:class}. We note, however, that these functions depend on several additional unknown models, and may not necessarily provide efficiency gain if one or more of these additional models is misspecified, even if $g_A$ and $g_Y$ are correct. The additional variability introduced by these parameters must be accounted for in finding a suitable variance estimator $\hat{\Omega}_n$ for both testing and estimation. This can be accomplished by stacking into $U(\psi,\alpha)$ the score or estimating equations used to estimate the nuisance parameters that estimates of the functions $\ell$ and $m$ depend on.

\section{Extensions to Binary and Count Exposures}\label{sec:exten}

Under stronger conditions, the test statistics described in this paper can be extended to binary- and count-exposure settings. We will now assume there is no interaction between $C$ and $X^*$ in the propensity-score model and that the measurement error is independent of $X^*$, $C$, $A$, and $Y$. Assuming the propensity-score model
\begin{align}\label{binarymodel}
\text{logit Pr}(A=1\mid C,X^*)=g_A(C;\alpha_1)+\alpha_2^T X^*,
\end{align}
for binary $A$ or
\begin{align}\label{countmodel1}
\log E(A\mid C,X^*)=g_A(C;\alpha_1)+\alpha_2^T X^*
\end{align}
for count $A$ is correctly specified, we have the following analogous result.

\begin{Theorem}\label{thm:discrete}
Let $\ell(C)$ and $m(C)$ each be vector-valued functions of $C$ with linearly-independent elements and dimension $p+q$, and let
$U(\alpha)\equiv\left\{\ell(C)Y+m(C)\right\}\exp(-\alpha_2^TXA)[A-\mathrm{expit}\{\allowbreak g_A(\allowbreak C;\allowbreak\alpha_1)\}]$
for (\ref{binarymodel}) if $A$ is binary or
$U(\alpha)\equiv\left\{\ell(C)Y+m(C)\right\}[A-\exp\{g_A(C;\alpha_1)+\alpha_2^TX\}]$
for (\ref{countmodel1}) if $A$ is a count. Suppose $U(\alpha)$ is continuously differentiable, $\nabla_\alpha E\{U(\alpha)\}=E\{\nabla_{\alpha} U(\alpha)\}$, and $\Omega^{-1} E\{\nabla_{\alpha} U(\alpha)\}$ has full rank, where $\Omega=E\left\{U(\alpha)U(\alpha)^T\right\}$. Under H$_0$, the test statistic $\chi^2_{rps}\equiv\min\limits_{\alpha}n\hat{U}_n(\alpha)^T\hat{\Omega}_n^{-1}\hat{U}_n(\alpha)\xrightarrow{d} \chi^2_q$, for any $\hat{\Omega}_n\xrightarrow{p}\Omega$.
\end{Theorem}

The robust outcome-regression and doubly-robust tests in Section \ref{sec:dr} can also be extended to a count-exposure setting under (\ref{countmodel1}). As shown in the Appendix, if $\Delta_A(\alpha)$ is redefined as $A\exp\left\{-\alpha_2X\right\}-\exp\left\{g_A(C;\alpha_1)\right\}$, Theorem 3 holds as stated, and Theorem 2 holds if $U(\alpha,\gamma)$ is replaced by 
\[\left[\begin{array}{c}
k(C)\Delta_Y(\gamma)A\exp\left\{-\alpha_2X\right\}\\
S(\gamma)
\end{array}\right].\]
Unfortunately, we have no such extensions for (\ref{binarymodel}). Upon specifying a structural conditional-mean model for the causal effect of $A$ on $Y$, the moment functions for each of these tests can be easily adapted to form estimating equations for the average causal effect as was shown for the continuous-exposure case.

\section{Discussion}\label{sec:discuss}

We have developed a large class of statistics for the null hypothesis of no causal effect accounting for confounder classical measurement error. This work contributes to the literature on measurement error not only in causal inference, but also in the absence of external information by leveraging causal assumptions to produce a function of the observed data that behaves as an instrumental variable. The tests presented here do not require a causal model to be specified; they only require specification of a conditional mean model of the exposure, outcome, or both. The doubly-robust test only requires one of these models to be correctly specified. The only assumption required beyond those inherent to the causal inference framework (e.g., no unobserved confounding) is that the conditional mean of exposure is linear in the error-prone confounders, and that the part of the propensity-score model multiplying $X^*$ is not saturated in $C$. The latter condition can in fact be relaxed if H$_0$ is replaced by the sharp null, in which case $Y$ is formally an IV. The functions in the estimating equations involving $Y$ can then be nonlinear in $Y$, and hence the number of linearly-independent elements is no longer restricted by the number of possible covariate patterns in $C$.

Sargan tests behave as goodness-of-fit tests, such that when the appropriate models are correctly specified, the tests presented here are powered to detect whether H$_0$ fits the data. However, the tests are also powered to detect model misspecification, so even in the case where there is no causal effect, the tests may reject in case of model misspecification. Thus, when our tests reject, it is prudent to supplement them with a Sargan goodness-of-fit test for each model used as described in Section \ref{sec:data} in order to ensure the results are not due to poor model fit. Unfortunately, these tests are not useful for detecting nonlinearity in $X^*$.

As this method relies on $Y_0$ behaving as an IV for $X^*$, its performance naturally depends on the strength of this conditional association given $C$. Thus, when this association is weak, $Y_0$ will be a weak instrument, and our methods will have reduced power. However, in this case $X^*$ will be a weak confounder which may not need to be accounted for in any case, and hence there is a bias-variance trade-off to be considered when deciding whether to use this method. An empirical test of the conditional association between $Y$ and $X$ given $C$ can be used as a guideline, however we leave a formal treatment of this issue and development of a hybrid method with an unadjusted (for $X^*$) analysis as a potential avenue for future work.

In the multicity application, we tested for an effect of temperature on mortality while accounting for confounding by an error-prone measurement of PM2.5, and discovered evidence of an effect in New York, NY and Stamford, CT. While results from standard tests agreed with our findings in New York, test results in Stamford disagreed, suggesting these standard tests were biased toward the null in this case. In several other cities, the standard tests showed evidence of an effect, while our measurement-error-robust tests did not. This suggested a possible bias in the standard tests due to measurement error resulting in false positives, and that our method may have protected us against making such an error.

The work presented here is not without limitations. Though the tests presented are robust to measurement error of a subset of confounders, at least one true confounder must be measured correctly. While we have managed to avoid the use of parametric models, the assumption of linearity in the error-contaminated confounders could be unrealistic in certain settings. In our data application, no goodness-of-fit test rejected at an $\alpha$ level of 0.10 after adding sufficiently many basis functions, however we cannot be certain that the goodness-of-fit tests of the final propensity-score models were powered to detect nonlinearities in the error-contaminated confounders. Finally, we did not test for lagged effects of temperature, which could contribute to the effect of temperature on mortality.

One direction for future work would be to re-analyze these data with nearby cities with similar climates clustered into regions, as in \cite{schwartz2015projections}. This is sensible since many of the cities examined are quite close, and differences in test results are likely due to differences in power rather than effect size. This approach would greatly improve power to detect effects of temperature in entire regions. As it stands, the present analysis serves as a practical illustration of the application and interpretation of our method.

\newpage

\appendix

\section*{APPENDIX}

\begin{proof}[Proof of Theorem 1]Let $\bar{\alpha}$ be the true value of $\alpha$. Under $\mathcal{M}_A$,

\begin{align*}
E\{U(\bar{\alpha})\} = &E[\{\ell(C)Y+m(C)\}\{A - E(A\mid C,X^*) - g_{A,2}(C;\bar{\alpha}_2)^T\varepsilon^*\}]\\
 = &E[\{\ell(C)Y+m(C)\}A] - E[E\{\ell(C)Y+m(C)\mid C,X^*\}E(A\mid C,X^*)] \\
& - E[\{\ell(C)Y+m(C)\}g_{A,2}(C;\bar{\alpha}_2)^T]E(\varepsilon^*)\\
 = &E[\{\ell(C)Y+m(C)\}A] - E\left\{E\left(E[\{\ell(C)Y+m(C)\}A\mid A,C,X^*]\mid C,X^*\right)\right\}\\
 = &0,
\end{align*}
since $\varepsilon^*$ is mean independent of $Y$ and $C$, and $E(Y\mid A,C,X^*)=E(Y\mid C,X^*)$, which is implied by H$_0$ and Assumption 1. The regularity conditions on $U(\alpha)$ are sufficient to ensure that $\bar{\alpha}$ is a local minimum of $n\hat{U}_n(\alpha)^T\hat{\Omega}_n^{-1}\hat{U}_n(\alpha)$, and hence $\alpha$ is locally identified under H$_0$. Then because $\mathrm{dim}\{U(\alpha)\}=\mathrm{dim}(\alpha)+q$, $E\{U(\alpha)\}=0$ is an overidentified moment restriction, and the statistic $\chi^2_{rps}$ has a limiting distribution of $\chi^2_q$.
\end{proof}
\begin{proof}[Proof of Theorem 2] Let $\bar{\gamma}$ and $\bar{\alpha}_2$ be the true values of $\gamma$ and $\alpha_2$, respectively. Under $\mathcal{M}_Y$, $E\{S(\bar{\gamma})\}=0$ and
\begin{align*}
&E\left[k(C)\left\{Y-g_Y(C;\bar{\gamma})\right\}\left\{A-g_{A,2}(C;\bar{\alpha}_2)^TX\right\}\right]\\
=&E\left[k(C)\left\{Y-E\left(Y\mid C\right)\right\}\left\{A-E(A\mid C,X^*)+E(A\mid C,X^*=0)-g_{A,2}(C;\bar{\alpha}_2)^T\varepsilon^*\right\}\right]\\
=&E\left[k(C)\left\{Y-E\left(Y\mid C\right)\right\}A\right]-E\left[k(C)E\left\{Y-E\left(Y\mid C\right)\mid C,X^*\right\}E(A\mid C,X^*)\right]\\
&+E\left[k(C)E\left\{Y-E(Y\mid C)\mid C\right\}\left\{E(A\mid C,X^*=0)-g_{A,2}(C;\bar{\alpha}_2)^TE(\varepsilon^*)\right\}\right]\\
=&E\left[k(C)\left\{Y-E\left(Y\mid C\right)\right\}A\right]-E\left\{k(C)E\left(E[\{Y-E\left(Y\mid C\right)\}A\mid A,C,X^*]\mid C,X^*\right)\right\}\\
=&0,
\end{align*}
since $\varepsilon^*$ is mean independent of $Y$ and $C$, and $E(Y\mid A,C,X^*)=E(Y\mid C,X^*)$, which is implied by H$_0$ and Assumption 1. The regularity conditions on $U(\alpha_2,\gamma)$ are sufficient to ensure that $(\bar{\alpha}_2,\bar{\gamma})$ is a local minimum of $n\hat{U}_n(\alpha_2,\gamma)^T\hat{\Omega}_n^{-1}\hat{U}_n(\alpha_2,\gamma)$, and hence $(\alpha_2,\gamma)$ is locally identified under H$_0$. Then because $\mathrm{dim}\{U(\alpha_2,\gamma)\}=\mathrm{dim}(\alpha_2)+\mathrm{dim}(\gamma)+q$, $E\{U(\alpha_2,\gamma)\}=0$ is an overidentified moment restriction, and the statistic $\chi^2_{ror}$ has a limiting distribution of $\chi^2_q$.
\end{proof}
\begin{proof}[Proof of Theorem 3]Let $\bar{\gamma}$ be the true value of $\gamma$ under $\mathcal{M}_Y$ and $\bar{\alpha}$ be the true value of $\alpha$ under $\mathcal{M}_A$. Under $\mathcal{M}_A$, there exists some $\tilde{\gamma}$ such that $E\{S(\tilde{\gamma})\}=0$. That
\[E\left[\begin{array}{c}
k(C)\Delta_Y(\tilde{\gamma})\Delta_A(\bar{\alpha})\\
\left\{\ell(C)Y+m(C)\right\}\Delta_A(\bar{\alpha})
\end{array}\right]=0\]
follows from the unbiasedness of $U(\alpha)$ shown in the proof of Theorem 1.

Under $\mathcal{M}_Y$, $E\{S(\bar{\gamma})\}=0$ and for any $\alpha_1$,
\begin{align*}
&E\left[k(C)\left\{Y-g_Y(C;\bar{\gamma})\right\}\left\{A-g_{A,1}(C;\alpha_1)-g_{A,2}(C;\bar{\alpha}_2)^TX\right\}\right]\\
=&E\left[k(C)\left\{Y-E\left(Y\mid C\right)\right\}\left\{A-E(A\mid C,X^*)+E(A\mid C,X^*=0)-g_{A,1}(C;\alpha_1)-g_{A,2}(C;\bar{\alpha}_2)^T\varepsilon^*\right\}\right]\\
=&E\left[k(C)\left\{Y-E\left(Y\mid C\right)\right\}A\right]-E\left[k(C)E\left\{Y-E\left(Y\mid C\right)\mid C,X^*\right\}E(A\mid C,X^*)\right]\\
&+E\left[k(C)E\left\{Y-E(Y\mid C)\mid C\right\}\left\{E(A\mid C,X^*=0)-g_{A,1}(C;\alpha_1)-g_{A,2}(C;\bar{\alpha}_2)^TE(\varepsilon^*)\right\}\right]\\
=&E\left[k(C)\left\{Y-E\left(Y\mid C\right)\right\}A\right]-E\left\{k(C)E\left(E[\{Y-E\left(Y\mid C\right)\}A\mid A,C,X^*]\mid C,X^*\right)\right\}\\
=&0,
\end{align*}
since $\varepsilon^*$ is mean independent of $Y$ and $C$, and $E(Y\mid A,C,X^*)=E(Y\mid C,X^*)$, which is implied by H$_0$ and Assumption 1. Finally, there exists some $\tilde{\alpha}_1$ such that
\[E\left[\left\{\ell(C)Y+m(C)\right\}\Delta_A(\tilde{\alpha}_1,\bar{\alpha}_2)\right]=0.\]

Thus, under any law in $\mathcal{M}_\cup$, $E\{U(\alpha,\gamma)\}=0$ has a solution under H$_0$. The regularity conditions on $U(\alpha,\gamma)$ are sufficient to ensure that $(\bar{\alpha},\tilde{\gamma})$ is a local minimum of $n\hat{U}_n(\alpha,\gamma)^T\hat{\Omega}_n^{-1}\hat{U}_n(\alpha,\gamma)$ under $\mathcal{M}_A$ and $(\alpha_0^*,\tilde{\alpha}_1,\bar{\alpha}_2,\bar{\gamma})$ is a local minimum under $\mathcal{M}_Y$, and hence $(\alpha,\gamma)$ is locally identified under $\mathcal{M}_\cup$ and H$_0$. Then because $\mathrm{dim}[U(\alpha,\gamma)]=\mathrm{dim}(\alpha)+\mathrm{dim}(\gamma)+q$, $E\{U(\alpha,\gamma)\}=0$ is an overidentified moment restriction, and the statistic $\chi^2_{dr}$ has a limiting distribution of $\chi^2_q$.

\end{proof}
\begin{proof}[Proof of Theorem 4]First, reparameterize the propensity-score model as
\[\text{logit Pr}(A=1\mid C,X^*)=\alpha_0+g_A(C;\alpha_1)+\alpha_2^T X^*,\]
for model (4) or
\[\log E(A\mid C,X^*)=\alpha_0+g_A(C;\alpha_1)+\alpha_2^T X^*,\]
for model (5), where $g_A(0;\alpha_1)=0$, such that $\alpha_0$ is a scalar intercept and $\alpha_1$ has dimension $p_1-1$. When $A$ is binary, for the true value $\bar{\alpha}$ of $\alpha$, we have
\begin{align*}
&E\left(\left\{\ell(C)Y+m(C)\right\}\exp \left( -\bar{\alpha }_2^TXA\right) \left[ A-\text{expit}\left\{\alpha_0^*+g_A(C;\bar{\alpha}_1)\right\} \right] \right)  \\
=&E\left(\left\{\ell(C)Y+m(C)\right\}\exp \left( -\bar{\alpha}_2^TX^*A-\bar{\alpha}_2^T\varepsilon^* A\right) \left[ A-\text{expit}\left\{\alpha_0^*+g_A(C;\bar{\alpha}_1)\right\} \right] \right)  \\
=&E\left(\left\{\ell(C)Y+m(C)\right\}\exp \left( -\bar{\alpha}_2^TX^*A\right) 
E\left\{ \exp \left( -\bar{\alpha}_2^T\varepsilon^* A\right) \mid A,Y,C,X^*\right\}\right. \\
&\left.\times\left[ A-\text{expit}\left\{\alpha_0^*+g_A(C;\bar{\alpha}_1)\right\} \right] \right)  \\
=&E\left(\left\{\ell(C)Y+m(C)\right\}\exp \left( -\bar{\alpha}_2^TX^*A\right) \exp
\left( KA\right) \left[ A-\text{expit}\left\{\alpha_0^*+g_A(C;\bar{\alpha}_1)\right\} \right] \right) 
\end{align*}%
where $\exp \left( K\right) =E\left\{ \exp \left( -\bar{\alpha}_2^T\varepsilon^* \right) \right\} $ is the moment generating function of $\varepsilon^*$ evaluated at $-\bar{\alpha}_2$. We then note that the joint
density of $\left( A,X^*\right) $ given $C$ can be expressed as
\[
f\left( A,X^*\mid C\right) =\frac{f\left( X^*\mid A=0,C\right) \exp \left( \bar{\alpha}_{2}^TX^*A\right) f\left( A\mid X^*=0,C\right) }{t(C)}
\]%
where $t(C)$ is a normalizing constant. We then have that 
\begin{align*}
&E\left\{\left\{\ell(C)Y+m(C)\right\}\exp \left( -\bar{\alpha}_2^TX^*A\right) \exp
\left( KA\right) \left[ A-\text{expit}\left(\alpha_0^*+g(C;\bar{\alpha}_1)\right) \right] \right\}  \\
=&E\int_x \sum_{a}\frac{f\left( x|A=0,C\right) \exp \left( \bar{\alpha}_2^Txa\right) f\left( a|X=0,C\right) }{t(C)} \\
&\times\left[\ell(C)E\left\{Y\mid a,x,C\right\}+m(C)\right]\exp \left(
-\bar{\alpha}_2^Txa\right) \exp \left( Ka\right) \left[ a-\text{expit}\left(
\alpha_0^*+g(C;\bar{\alpha}_1)\right) \right]dx \\
=&E\int_x \left[\ell(C)E\left\{Y\mid x,C\right\}+m(C)\right]f\left(
x|A=0,C\right) t(C)^{-1}dx \\
&\times \sum_{a}f\left( a|X=0,C\right) \exp \left( Ka\right) \left[ a-\text{%
expit}\left(\alpha_0^*+g_A(C;\bar{\alpha}_1)\right) \right]  \\
=&E\int_x\left[\ell(C)E\left\{Y\mid x,C\right\}+m(C)\right]f\left(
x|A=0,C\right) t(C)^{-1}dx\frac{1+\exp \left(\alpha_0^*+g_A(C;\bar{\alpha}_1)\right) }{1+\exp \left(\bar{\alpha}_0+g_A(C;\bar{\alpha}_1)\right) } \\
&\times \sum_{a}\frac{\exp \left(\alpha_0^*a+g_A(C;\bar{\alpha}_1)a\right) }{1+\exp \left(\alpha_0^*+g_A(C;\bar{\alpha}_1)\right) }\left[ a-\text{expit}\left(\alpha_0^*+g_A(C;\bar{\alpha}_1)\right) \right]\\
=&0
\end{align*}%
where $\alpha _{0}^{\ast }=K+\bar{\alpha}_{0}$.

When $A$ is a count, let $\exp(K)=E\left\{\exp(\bar{\alpha}_2^T\varepsilon^*)\right\}$, i.e., the moment generating function of $\varepsilon^*$ evaluated at $\bar{\alpha}_2$. Under model (5), for the true value $\bar{\alpha}$ of $\alpha$, we have
\begin{align*}
&E\left(\left\{\ell(C)Y+m(C)\right\}\left[A-\exp\left\{\alpha_0^*+g_A(C;\bar{\alpha}_1)+\bar{\alpha}_2^TX\right\}\right]\right)\\
=&E\left[\left\{\ell(C)Y+m(C)\right\}A\right]-E\left[\left\{\ell(C)Y+m(C)\right\}\exp\left\{\alpha_0^*+g_A(C;\bar{\alpha}_1)+\bar{\alpha}_2^TX^*\right\}\right.\\
&\times \left.E\left\{\exp(\bar{\alpha}_2^T\varepsilon^*)\right\}\right]\\
=&E\left[\left\{\ell(C)Y+m(C)\right\}A\right]-E\left[\left\{\ell(C)Y+m(C)\right\}\exp\left\{\bar{\alpha}_0+g_A(C;\bar{\alpha}_1)+\bar{\alpha}_2^TX^*\right\}\right]\\
=&E\left[\left\{\ell(C)Y+m(C)\right\}A\right]-E\left[\left\{\ell(C)E(Y\mid C,X^*)+m(C)\right\}E(A\mid C,X^*)\right]\\
=&E\left[\left\{\ell(C)Y+m(C)\right\}A\right]-E\left(E\left[\left\{\ell(C)E(Y\mid C,X^*,A)+m(C)\right\}A\mid C,X^*\right]\right)\\
=&E\left[\left\{\ell(C)Y+m(C)\right\}A\right]-E\left[E\left\{\left(E\left[\left\{\ell(C)Y+m(C)\right\}A\mid C,X^*,A\right]\right)\mid C,X^*\right\}\right]\\
=&0,
\end{align*}
where $\alpha _{0}^{\ast }=K+\bar{\alpha}_{0}$.

Thus, in either case, $E\{U(\alpha_0^*,\bar{\alpha}_1,\bar{\alpha}_2)\}=0$ under H$_0$. The regularity conditions on $U(\alpha)$ are sufficient to ensure that $(\alpha_0^*,\bar{\alpha}_1,\bar{\alpha}_2)$ is a local minimum of $n\hat{U}_n(\alpha)^T\hat{\Omega}_n^{-1}\hat{U}_n(\alpha)$, and hence $\alpha$ is locally identified under H$_0$. Then because $\mathrm{dim}\{U(\alpha)\}=\mathrm{dim}(\alpha)+q$, $E\{U(\alpha)\}=0$ is an overidentified moment restriction, and the statistic $\chi^2_{\mathrm{robust}_A}$ has a limiting distribution of $\chi^2_q$.
\end{proof}

\begin{Theorem}\label{thm:eff}
Let $\hat{\beta}(\ell,m)$ be the estimator solving $\mathbb{P}_nU(\ell,m;\beta)=0$ corresponding to the moment functions in Theorem 1, and define \[
d(C)\equiv E\{\Delta(\alpha)^2H(\psi)^2\mid C\}E\{\Delta(\alpha)^2\mid C\}-E\{\Delta(\alpha)^2H(\psi)\mid C\}E\{\Delta(\alpha)^2H(\psi)\mid C\},
\]

\[\ell^*(C)\equiv d(C)^{-1}\left[\begin{array}{c}
E\left\{\Delta(\alpha)^2\mid C\right\}E\left\{\Delta(\alpha)A\mid C\right\}\\
-\mathrm{Cov}\{\Delta(\alpha)^2,H(\psi)\mid C\}\nabla_{\alpha_1} g_{A,1}(C;\alpha_1)\\
\left[E\left\{\Delta(\alpha)^2\mid C\right\}E(H(\psi)X^T\mid C)-E\left\{\Delta(\alpha)^2H(\psi)\mid C\right\}E(X^T\mid C)\right]\\
\times\nabla_{\alpha_2} g_{A,2}(C;\alpha_2)\\
\end{array}\right],\]
and
\[m^*(C)\equiv d(C)^{-1}\left[\begin{array}{c}
-E\left\{\Delta(\alpha)^2H(\psi)\mid C\right\}E\left\{\Delta(\alpha)A\mid C\right\}\\
\mathrm{Cov}\{\Delta(\alpha)^2H(\psi),H(\psi)\mid C\}\nabla_{\alpha_1} g_{A,1}(C;\alpha_1)\\
-\left[E\left\{\Delta(\alpha)^2H(\psi)\mid C\right\}E(H(\psi)X^T\mid C) - E\left\{\Delta(\alpha)^2H(\psi)^2\mid C\right\}E(X^T\mid C)\right]\\
\times\nabla_{\alpha_2} g_{A,2}(C;\alpha_2)
\end{array}\right].\]
Under $\mathcal{M}_A$ and the causal model defined by equation (3), $\hat{\beta}(\ell^*,m^*)$ achieves the minimum asymptotic variance of all estimators in the class of estimators defined by the estimating equations in Theorem 1. The corresponding variance is $\E\left[\mathbf{U}(\ell^*,m^*;\beta)\mathbf{U}(\ell^*,m^*;\beta)^T\right]$.
\end{Theorem}
\begin{proof}By Theorem 5.3 in \cite{newey1994large}, if an optimal estimator $\hat{\beta}(\tilde{\ell},\tilde{m})$ exists within the class $\{\hat{\beta}(\ell,m):\; \ell\in\mathcal{L}, m\in \mathcal{M}\}$, the functions $\tilde{\ell}$ and $\tilde{m}$ are guaranteed to satisfy
\begin{align}
-\E\left[\frac{\partial}{\partial \beta} \mathbf{U}(\ell,m;\beta)\right] = \E\left[\mathbf{U}(\ell,m;\beta)\mathbf{U}(\tilde{\ell},\tilde{m};\beta)^T\right]
\end{align}
for all functions $\ell$ and $m$, and the estimator will have variance equal to $\E\left[\mathbf{U}(\tilde{\ell},\tilde{m};\beta)\mathbf{U}(\tilde{\ell},\tilde{m};\beta)^T\right]$. Thus it suffices to show that $\ell^*(C)$ and $m^*(C)$ satisfy (2). We have
\begin{align*}
-E\left[\frac{\partial}{\partial \beta} \mathbf{U}(\ell,m;\beta)\right] =& E\left[\ell(C)A\Delta(\alpha), \;\; \{\ell(C)H(\psi)+m(C)\} \nabla_{\alpha_1} g_{A,1}(C;\alpha_1),\right.\\
&\left. \{\ell(C)H(\psi)+m(C)\}X^T\nabla_{\alpha_2} g_{A,2}(C;\alpha_2)\right]\\
&= E\left\{\left[\ell(C), m(C)\right]\left[
\begin{array}{c}
V\\
W
\end{array}\right]\right\},
\end{align*}
where $V=[V_1,V_2,V_3]\equiv [ A\Delta(\alpha), H(\psi)\nabla_{\alpha_1} g_{A,1}(C;\alpha_1), H(\psi)X^T\nabla_{\alpha_2} g_{A,2}(C;\alpha_2)]$ and $W=[W_1,W_2,W_3]\equiv [0, \nabla_{\alpha_1} g_{A,1}(C;\alpha_1), X^T\nabla_{\alpha_2} g_{A,2}(C;\alpha_2)]$. If we partition the components of the functions $(\ell^*,m^*)$ into $(\ell_1^*,m_1^*)$, $(\ell_2^*,m_2^*)$, and $(\ell_3^*,m_3^*)$, where $\ell_1^*$ and $m_1^*$ are scalar functions, $\ell_2^*$ and $m_2^*$ are $p_1$ dimensional, and $\ell_3^*$ and $m_3^*$ are $p_2$ dimensional, then
\begin{align*}
\E&\left[\mathbf{U}(\ell,m;\beta)\mathbf{U}(\ell^*,m^*;\beta)^T\right] = E\left[\Delta(\alpha)^2 \{\ell(C)H(\psi)+m(C)\}\{\ell^*(C)H(\psi)+m^*(C)\}^T\right]\\
= &E\left[\Delta(\alpha)^2\{\ell(C)H(\psi)+m(C)\}\left[\ell_1^*(C)H(\psi)+m_1^*(C), \;\; \{\ell_2^*(C)H(\psi)+m_2^*(C)\}^T, \right.\right.\\
&\left.\left.\{\ell_3^*(C)H(\psi)+m_3^*(C)\}^T\right]\right].
\end{align*}

Thus, we can solve (1) by partitioning it into four independent equations corresponding to the partition of $\ell^*$ and $m^*$: For all $k\in \{1,2,3\}$,
\[E\left[\Delta(\alpha)^2\{\ell(C)H(\psi)+m(C)\}\left\{\ell_k^*(C)H(\psi)+m_k^*(C)\right\}\right] = E\left\{\ell(C)V_k+ m(C)W_k\right\}.\]
\begin{align*}
E\left[ E\left[\Delta(\alpha)^2\left\{\ell_k^*(C)H(\psi)+m_k^*(C)\right\}H(\psi)\mid C\right] \ell(C)\right.&\\
\left.-E\left[\Delta(\alpha)^2\left\{\ell_k^*(C)H(\psi)+m_k^*(C)\right\}\mid C\right] m(C)\right]&=E\left\{\ell(C)V_k+ m(C)W_k\right\}
\end{align*}
$\Leftrightarrow$
\begin{align*}
E\left[ E\left[\Delta(\alpha)^2\left\{\ell_k^*(C)H(\psi)+m_k^*(C)\right\}H(\psi)-V_k\mid C\right] \ell(C)\right.&\\
\left.-E\left[\Delta(\alpha)^2\left\{\ell_k^*(C)H(\psi)+m_k^*(C)\right\}-W_k\mid C\right] m(C)\right]&=0
\end{align*}
$\Leftrightarrow$
\begin{align*}
E[\Delta(\alpha)^2H(\psi)^2\mid C]\ell_k^*(C) + E[\Delta(\alpha)^2H(\psi)\mid C]m_k^*(C) - E(V_k\mid C) &= 0\\
E[\Delta(\alpha)^2H(\psi)\mid C]\ell_k^*(C) + E[\Delta(\alpha)^2\mid C]m_k^*(C) - E(W_k\mid C) &= 0
\end{align*}
$\Leftrightarrow$
\[\left[\begin{array}{c}
\ell_k^*(C)\\
m_k^*(C)
\end{array}\right] = 
\left[\begin{array}{cc}
E\{\Delta(\alpha)^2H(\psi)^2\mid C\} & E\{\Delta(\alpha)^2H(\psi)\mid C\}\\
E\{\Delta(\alpha)^2H(\psi)\mid C\} & E\{\Delta(\alpha)^2\mid C\}
\end{array}\right]^{-1}
\left[\begin{array}{c}
E(V_k\mid C)\\
E(W_k\mid C)
\end{array}\right]\]
given that $\text{Pr}\{d(C)= 0\}=0$. The second implication can be seen to hold by recognizing the necessity of the first equation when $\ell(C)=E\left[\Delta(\alpha)^2\left\{\ell_k^*(C)H(\psi)+m_k^*(C)\right\}H(\psi)-V_k\mid C\right]$ and $m(C)=0$ and the necessity of the second equation when $\ell(C)=0$ and \[m(C)=E\left[\Delta(\alpha)^2\left\{\ell_k^*(C)H(\psi)+m_k^*(C)\right\}-W_k\mid C\right].\] Thus, (2) is solved by
\begin{align*}
\ell_1^*(C) =& d(C)^{-1}E\left\{\Delta(\alpha)^2\mid C\right\}E\left\{\Delta(\alpha)A\mid C\right\}\\
m_1^*(C) =& -d(C)^{-1}E\left\{\Delta(\alpha)^2H(\psi)\mid C\right\}E\left\{\Delta(\alpha)A\mid C\right\}\\
\ell_2^*(C) =& -d(C)^{-1}\mathrm{Cov}\{\Delta(\alpha)^2,H(\psi)\mid C\}\nabla_{\alpha_1} g_{A,1}(C;\alpha_1)\\
m_2^*(C) =& d(C)^{-1}\mathrm{Cov}\{\Delta(\alpha)^2H(\psi),H(\psi)\mid C\}\nabla_{\alpha_1} g_{A,1}(C;\alpha_1)\\
\ell_3^*(C) =& d(C)^{-1}\left[E\left\{\Delta(\alpha)^2\mid C\right\}E(H(\psi)X^T\mid C)\nabla_{\alpha_2} g_{A,2}(C;\alpha_2)\right.\\
&\left. - E\left\{\Delta(\alpha)^2H(\psi)\mid C\right\}E(X^T\mid C)\nabla_{\alpha_2} g_{A,2}(C;\alpha_2)\right]\\
m_3^*(C) =& -d(C)^{-1}\left[E\left\{\Delta(\alpha)^2H(\psi)\mid C\right\}E(H(\psi)X^T\mid C)\nabla_{\alpha_2} g_{A,2}(C;\alpha_2)\right.\\
& \left. - E\left\{\Delta(\alpha)^2H(\psi)^2\mid C\right\}E(X^T\mid C)\nabla_{\alpha_2} g_{A,2}(C;\alpha_2)\right].
\end{align*}
\end{proof}

\begin{proof}[Proof of validity of the robust outcome-regression test under model (5)]
Let $\bar{\gamma}$ and $\bar{\alpha}_2$ be the true values of $\gamma$ and $\alpha_2$ respectively. Under $\mathcal{M}_Y$, $E\{S(\bar{\gamma})\}=0$ and
\begin{align*}
&E\left[k(C)\left\{Y-g_Y(C;\bar{\gamma})\right\}A\exp(-\bar{\alpha}_2^TX)\right]\\
=&E\left[k(C)E\left\{Y-E\left(Y\mid C\right)\mid C,X^*,A\right\}A\frac{E(A\mid C,X^*=0)}{E(A\mid C,X^*)}E\left\{\exp(-\bar{\alpha}_2^T\varepsilon^*)\right\}\right]\\
=&E\left[k(C)E\left\{Y-E\left(Y\mid C\right)\mid C,X^*\right\}A\frac{E(A\mid C,X^*=0)}{E(A\mid C,X^*)}\right]e^K\\
=&E\left[k(C)E\left\{Y-E\left(Y\mid C\right)\mid C,X^*\right\}E(A\mid C,X^*)\frac{E(A\mid C,X^*=0)}{E(A\mid C,X^*)}\right]e^K\\
=&E\left[k(C)E\left\{Y-E\left(Y\mid C\right)\mid C\right\}E(A\mid C,X^*=0)\right]e^K\\
=&0,
\end{align*}
where $e^K=E\left\{\exp(-\bar{\alpha}_2^T\varepsilon^*)\right\}$, since $\varepsilon^* \ci \{X^*,C,A,Y\}$, and $E(Y\mid A,C,X^*)=E(Y\mid C,X^*)$, which is implied by H$_0$ and Assumption 1. The regularity conditions on $U(\alpha_2,\gamma)$ are sufficient to ensure that $(\bar{\alpha}_2,\bar{\gamma})$ is a local minimum of $n\hat{U}_n(\alpha_2,\gamma)^T\hat{\Omega}_n^{-1}\hat{U}_n(\alpha_2,\gamma)$, and hence $(\alpha_2,\gamma)$ is locally identified under H$_0$. Then because $\mathrm{dim}\{U(\alpha_2,\gamma)\}=\mathrm{dim}(\alpha_2)+\mathrm{dim}(\gamma)+q$, $E\{U(\alpha_2,\gamma)\}=0$ is an overidentified moment restriction, and the statistic $\chi^2_{ror}$ has a limiting distribution of $\chi^2_q$.
\end{proof}

\begin{proof}[Proof of validity of the doubly-robust test under model (5)]
First, reparameterize the propensity-score model as \[\log E(A\mid C,X^*)=\alpha_0+g_A(C;\alpha_1)+\alpha_2^T X^*,\]
where $g_A(0;\alpha_1)=0$, such that $\alpha_0$ is a scalar intercept and $\alpha_1$ has dimension $p_1-1$. Let $\exp(K)=E\left\{\exp(\bar{\alpha}_2^T\varepsilon^*)\right\}$, i.e., the moment generating function of $\varepsilon^*$ evaluated at $\bar{\alpha}_2$, $\bar{\gamma}$ be the true value of $\gamma$ under $\mathcal{M}_Y$, and $\bar{\alpha}$ be the true value of $\alpha$ under $\mathcal{M}_A$. Under $\mathcal{M}_A$, there exists some $\tilde{\gamma}$ such that $E\{S(\tilde{\gamma})\}=0$, and for any such $\tilde{\gamma}$,
\begin{align*}
&E\left[\begin{array}{c}
k(C)\Delta_Y(\tilde{\gamma})\Delta_A(\bar{\alpha})\\
\left\{\ell(C)Y+m(C)\right\}\Delta_A(\bar{\alpha})
\end{array}\right]\\
=&E\left(\left[\begin{array}{c}k(C)\left\{Y-g_Y(C;\tilde{\gamma})\right\}\\
\ell(C)Y+m(C)\end{array}\right]\left[A\exp(-\bar{\alpha}_2^TX)-\exp\left\{\alpha_0^*+g_A(C;\bar{\alpha}_1)\right\}\right]
\right)\\
=&E\Biggl(\left[\begin{array}{c}k(C)\left\{E(Y\mid C,X^*,A)-g_Y(C;\tilde{\gamma})\right\}\\
\ell(C)E(Y\mid C,X^*,A)+m(C)\end{array}\right]\\
&\times\left[A\frac{E(A\mid C,X^*=0)}{E(A\mid C,X^*)}E\left\{\exp(-\bar{\alpha}_x^T\varepsilon^*)\right\}-\exp\left\{K+\bar{\alpha}_0+g_A(C;\bar{\alpha}_1)\right\}\right]
\Biggr)\\
=&e^KE\left(\left[\begin{array}{c}k(C)\left\{E(Y\mid C,X^*)-g_Y(C;\tilde{\gamma})\right\}\\
\ell(C)E(Y\mid C,X^*)+m(C)\end{array}\right]\left[A\frac{E(A\mid C,X^*=0)}{E(A\mid C,X^*)}-E(A\mid C,X^*=0)\right]
\right)\\
=&e^KE\Biggl(\left[\begin{array}{c}k(C)\left\{E(Y\mid C,X^*)-g_Y(C;\tilde{\gamma})\right\}\\
\ell(C)E(Y\mid C,X^*)+m(C)\end{array}\right]\\
&\times\left[E(A\mid C,X^*)\frac{E(A\mid C,X^*=0)}{E(A\mid C,X^*)}-E(A\mid C,X^*=0)\right]
\Biggr)\\
=&0.
\end{align*}
Under $\mathcal{M}_Y$, $E\left\{S(\bar{\gamma})\right\}=0$, and for any $\alpha_1$,
\begin{align*}
&E\left[k(C)\Delta_Y(\tilde{\gamma})\Delta_A(\alpha_1,\bar{\alpha_2})\right]\\
=&E\left(k(C)\left\{Y-E(Y\mid C)\right\}\left[A\exp(-\bar{\alpha}_2^TX)-\exp\left\{g_A(C;\alpha_1)\right\}\right]\right)\\
=&E\left(k(C)E\left\{Y-E(Y\mid C)\mid C,X^*,A\right\}\left[A\frac{E(Y\mid C,X^*=0)}{E(Y\mid C,X^*)}\exp(-\bar{\alpha}_2^T\varepsilon^*)-\exp\left\{g_A(C;\alpha_1)\right\}\right]\right)\\
=&E\left(k(C)E\left\{Y-E(Y\mid C)\mid C,X^*\right\}\left[A\frac{E(Y\mid C,X^*=0)}{E(Y\mid C,X^*)}e^K-\exp\left\{g_A(C;\alpha_1)\right\}\right]\right)\\
=&E\left(k(C)E\left\{Y-E(Y\mid C)\mid C,X^*\right\}\left[E(Y\mid C,X^*)\frac{E(Y\mid C,X^*=0)}{E(Y\mid C,X^*)}e^K-\exp\left\{g_A(C;\alpha_1)\right\}\right]\right)\\
=&E\left(k(C)E\left\{Y-E(Y\mid C)\mid C\right\}\left[E(Y\mid C,X^*=0)e^K-\exp\left\{g_A(C;\alpha_1)\right\}\right]\right)\\
=&0,
\end{align*}
since $\varepsilon^* \ci \{X^*,C,A,Y\}$ and $E(Y\mid A,C,X^*)=E(Y\mid C,X^*)$, which is implied by H$_0$ and Assumption 1. Finally, there exists some $\tilde{\alpha}_1$ such that
\[E\left[\left\{\ell(C)Y+m(C)\right\}\Delta_A(\alpha_0^*,\tilde{\alpha}_1,\bar{\alpha}_2)\right]=0.\]
Thus, under any law in $\mathcal{M}_\cup$, $E\{U(\alpha,\gamma)\}=0$ has a solution under H$_0$. The regularity conditions on $U(\alpha,\gamma)$ are sufficient to ensure that $(\bar{\alpha},\tilde{\gamma})$ is a local minimum of $n\hat{U}_n(\alpha,\gamma)^T\hat{\Omega}_n^{-1}\hat{U}_n(\alpha,\gamma)$ under $\mathcal{M}_A$ and $(\alpha_0^*,\tilde{\alpha}_1,\bar{\alpha}_2,\bar{\gamma})$ is a local minimum under $\mathcal{M}_Y$, and hence $(\alpha,\gamma)$ is locally identified under $\mathcal{M}_\cup$ and H$_0$. Then because $\mathrm{dim}[U(\alpha,\gamma)]=\mathrm{dim}(\alpha)+\mathrm{dim}(\gamma)+q$, $E\{U(\alpha,\gamma)\}=0$ is an overidentified moment restriction, and the statistic $\chi^2_{dr}$ has a limiting distribution of $\chi^2_q$.
\end{proof}

\newpage

\bibliographystyle{apalike}

\bibliography{references}

\end{document}